\documentclass[11pt]{article}

\usepackage{amsmath,amssymb,amsthm,mathtools}
\usepackage{xcolor}
\usepackage[colorlinks=true,linkcolor=blue,citecolor=blue,urlcolor=blue]{hyperref}
\usepackage[nameinlink,capitalize]{cleveref}

\newtheorem{theorem}{Theorem}
\newtheorem{lemma}[theorem]{Lemma}
\newtheorem{corollary}[theorem]{Corollary}
\theoremstyle{definition}
\newtheorem{definition}{Definition}
\theoremstyle{plain}

\newcommand{\F}{\mathbb F}
\newcommand{\E}{\mathbb E}
\newcommand{\ket}[1]{|#1\rangle}
\newcommand{\bra}[1]{\langle #1|}

\newcommand{\Inc}{\operatorname{Inc}}
\newcommand{\GL}{\mathrm{GL}}
\newcommand{\Tot}{\deg_{\mathrm{tot}}}

\newcommand{\polyind}[3]{\mathrm{Poly}_{\mathrm{ind}}(#1,#2,#3)}
\newcommand{\polytot}[3]{\mathrm{Poly}_{\mathrm{tot}}(#1,#2,#3)}

\newcommand{\polymeas}[3]{\mathrm{PolyMeas}(#1,#2,#3)}
\newcommand{\TotPoly}{\operatorname{TotPoly}}
\newcommand{\TotPolyMeas}{\operatorname{TotPolyMeas}}

\newcommand{\bu}{\boldsymbol{u}}
\newcommand{\bv}{\boldsymbol{v}}

\newcommand{\by}{\boldsymbol{y}}
\newcommand{\bi}{\boldsymbol{i}}

\newcommand{\LinePt}{\mathsf{LinePt}}
\newcommand{\DiagPt}{\mathsf{DiagPt}}
\newcommand{\Bad}{\mathcal B}

\begin{document}

\title{Quantum Soundness of a Total-Degree Line-versus-Point Test}
\author{Honghao Fu \and Nengkun Yu \and Tianrun Zhao}
\date{\today}
\maketitle

\begin{abstract}
We prove quantum soundness of the total-degree diagonal line-vs-point test
using the individual-degree soundness theorem of Ji, Natarajan,
Vidick, Wright, and Yuen.  A random change of coordinates yields
projective polynomial decoders of total degree at most \(md\).
The uniform-line slice of the test bounds the weight of outcomes of degree
greater than \(d\), which are removed by a common relabeling.
This reduction does not yield a dimension-independent soundness bound:
the \(\operatorname{poly}(m)\) dependence of the individual-degree
theorem persists, as discussed in \cite[Section~1.2]{JNVWY20}.
\end{abstract}

\section{Introduction}

Low-degree tests ask whether local polynomial answers are consistent
with a single global polynomial.  For provers sharing entanglement,
the corresponding soundness statement concerns global measurements:
each prover should admit a polynomial-indexed measurement whose
evaluations agree with the other prover's point answers, and the two
decoded polynomial labels should agree with high probability.
Here we study this question when the test samples from the diagonal line
distribution defined below and requires answers of degree at most \(d\).

The proof separates constructing a global decoder from controlling its
total degree.  A random invertible change of coordinates allows us to
apply the individual-degree soundness theorem
of Ji--Natarajan--Vidick--Wright--Yuen~\cite[Theorem~3.10]{JNVWY20}.
Its decoders may have total degree as large as \(md\).  When \(q\) is
sufficiently large relative to \(md\), a polynomial of excessive total
degree disagrees with every permitted
line answer on most points of most lines.  The original test therefore
bounds the weight of these outcomes, which we relabel by the zero
polynomial.  This preserves projectivity and agreement between the
decoders.  The resulting short reduction inherits the quantitative
parameters, including the dimension dependence, of the cited theorem.

\section{Notation and conventions}

Let \(q\) be a prime power and let \(m\ge1\), \(d\ge0\) be integers.
Every function \(\F_q^m\to\F_q\) has a unique reduced representative,
meaning a polynomial whose degree in each variable is at most \(q-1\).
For such a function \(P\), write
\[
  \Tot(P)
\]
for the total degree of its reduced representative.

Write
\[
  \polytot{m}{q}{d}
\]
for the set of functions \(\F_q^m\to\F_q\) represented by reduced
polynomials of total degree at most \(d\), and write
\[
  \polyind{m}{q}{d}
\]
for the set of functions represented by reduced polynomials of individual
degree at most \(d\).
As in \cite{JNVWY20}, \(\polymeas{m}{q}{d}\) denotes the set
of measurements indexed by \(\polyind{m}{q}{d}\).

Let \(\LinePt\) denote the uniform nondegenerate point-line distribution:
sample \(\bu\in\F_q^m\) uniformly, then sample a uniformly random
nondegenerate affine line \(\ell\subseteq\F_q^m\) containing \(\bu\).

Let \(\DiagPt\) denote the diagonal/general-line distribution from
\cite[Section~3]{JNVWY20}: sample
\[
  \bu\sim\F_q^m,\qquad
  \bi\sim\{1,\ldots,m\},
\]
uniformly, then sample
\[
  \bv\in V_{\bi}
  :=
  \{\bv\in\F_q^m:\ v_{\bi+1}=\cdots=v_m=0\}
\]
uniformly, and set
\[
  \ell=\{\bu+t\bv:t\in\F_q\}.
\]
If \(\bv=0\), we regard \(\ell=\{\bu\}\) as a degenerate singleton line.
Thus this note uses the literal diagonal sampler, including the event
\(\bv=0\), and extends line queries to singleton queries. All samplers
below use this literal convention.

Throughout, \(r\in\{\mathrm A,\mathrm B\}\) indexes the two provers, who
share a normalized state
\(\ket{\psi}\in\mathcal H_{\mathrm A}\otimes\mathcal H_{\mathrm B}\).
Prover \(r\)'s measurement operators act on \(\mathcal H_r\).

If \(G=\{G_P\}_P\) is a measurement indexed by functions
\(P:\F_q^m\to\F_q\), define its point-value coarse-graining by
\[
  G^u_a
  :=
  \sum_{P:P(u)=a}G_P.
\]
Equivalently, we write
\[
  G_{[P(u)=a]}
  :=
  \sum_{P:P(u)=a}G_P.
\]
Similarly, if
\[
  L^\ell=\{L^\ell_h\}_h
\]
is a line measurement, define
\[
  L^\ell_{[h(u)=a]}
  :=
  \sum_{h:h(u)=a}L^\ell_h.
\]
We write \(L^{r,\ell,u}:=\{L^{r,\ell}_{[h(u)=a]}\}_{a\in\F_q}\).

We use the probability-style consistency convention from
\cite[Section~4]{JNVWY20}. Let \(\mathcal D\) be a probability
distribution on the set of question indices \(x\).
Then
\[
  M^x_a\otimes I\simeq_\eta I\otimes N^x_a
\]
on \(x\sim\mathcal D\) means
\[
  \E_{x\sim\mathcal D}
  \sum_{a\ne b}
  \bra{\psi}M^x_a\otimes N^x_b\ket{\psi}
  \le \eta.
\]
For measurements \(M,N\) on the two respective subsystems, write
\[
  \Inc_\psi(M,N)
  :=\sum_{a\ne b}\bra{\psi}M_a\otimes N_b\ket{\psi}.
\]
We use the deterministic post-processing fact for \(\simeq_\eta\)
from the preliminaries of
\cite[Section~4]{JNVWY20}.

\begin{definition}[Total-degree diagonal line-vs-point test]
\label{def:total-degree-diagonal-line-vs-point}
Let \(m\ge1\), let \(q\) be a prime power, and let \(d\ge0\).
The degree-\(d\) total-degree diagonal line-vs-point test is the following
two-prover test.

With probability \(1/3\) each, perform one of the following tests.

\begin{enumerate}
\item \textbf{Alice point, Bob diagonal/general line.}
Sample \((\bu,\bi,\bv,\ell)\sim\DiagPt\).
Send \(\bu\) to Alice and \(\ell\) to Bob.
Alice answers \(a\in\F_q\).
Bob answers a function \(h:\ell\to\F_q\) which is the restriction of a
univariate polynomial of degree at most \(d\) on \(\ell\).
Accept iff
\[
  h(\bu)=a.
\]

\item \textbf{Alice diagonal/general line, Bob point.}
Sample \((\bu,\bi,\bv,\ell)\sim\DiagPt\).
Send \(\ell\) to Alice and \(\bu\) to Bob.
Alice answers a function \(h:\ell\to\F_q\) which is the restriction of a
univariate polynomial of degree at most \(d\) on \(\ell\).
Bob answers \(a\in\F_q\).
Accept iff
\[
  h(\bu)=a.
\]

\item \textbf{Point self-consistency.}
Sample \(\bu\sim\F_q^m\) uniformly.
Send \(\bu\) to both provers.
They answer \(a,b\in\F_q\).
Accept iff
\[
  a=b.
\]
\end{enumerate}

For a singleton line \(\ell=\{\bu\}\), every function
\(h:\ell\to\F_q\) is regarded as degree \(0\), hence as degree at most
\(d\).

A projective strategy for this test is denoted
\[
  (\psi,A^{\mathrm A},L^{\mathrm A},A^{\mathrm B},L^{\mathrm B}),
\]
where \(A^{r,u}=\{A^{r,u}_a\}_{a\in\F_q}\) is a projective point
measurement, and \(L^{r,\ell}=\{L^{r,\ell}_h\}_{\deg h\le d}\) is a
projective line measurement.
\end{definition}

\begin{definition}[Total-degree polynomial measurements]
Let
\[
  \TotPoly(m,q,d)
\]
denote the set of functions \(\F_q^m\to\F_q\) represented by reduced
polynomials of total degree at most \(d\).
Let
\[
  \TotPolyMeas(m,q,d)
\]
denote the set of projective measurements whose outcomes are indexed by
\(\TotPoly(m,q,d)\).
\end{definition}

\section{Main result}

\begin{theorem}[Quantum soundness of the total-degree diagonal line-vs-point test]
\label{thm:total-degree-diagonal-from-individual-degree}
Let
\[
  (\psi,A^{\mathrm A},L^{\mathrm A},A^{\mathrm B},L^{\mathrm B})
\]
be a projective strategy that passes the degree-\(d\) total-degree
diagonal line-vs-point test with probability at least \(1-\epsilon\),
where \(0\le\epsilon\le1\).
Let \(k\ge md\), and put
\[
  c_0=\frac1{40000}.
\]
For some universal constant \(C_0\), put
\[
  \nu_{\mathrm{tot}}
  :=
  C_0 (k+1)^2m^4
  \left(
    \epsilon^{c_0}
    +
    \left(\frac{d}{q}\right)^{c_0}
    +
    e^{-k/(2560000m^2)}
  \right).
\]
Then there exist projective measurements
\[
  G^{\mathrm A},G^{\mathrm B}\in\TotPolyMeas(m,q,d)
\]
such that, on average over \(\bu\sim\F_q^m\),
\[
  A^{\mathrm A,u}_a\otimes I
  \simeq_{\nu_{\mathrm{tot}}}
  I\otimes G^{\mathrm B}_{[g(u)=a]},
\]
\[
  I\otimes A^{\mathrm B,u}_a
  \simeq_{\nu_{\mathrm{tot}}}
  G^{\mathrm A}_{[g(u)=a]}\otimes I,
\]
and
\[
  G^{\mathrm A}_g\otimes I
  \simeq_{\nu_{\mathrm{tot}}}
  I\otimes G^{\mathrm B}_g.
\]
\end{theorem}

We prove the theorem by first constructing decoders of total degree at
most \(md\), then comparing them with the original line measurements and
relabeling the outcomes of degree greater than \(d\).
The proof is assembled in \Cref{sec:total-degree-proof}.

\section{Constructing preliminary decoders}

\begin{lemma}[Reduction to the low-individual-degree test]
\label[lemma]{lem:random-coordinate-reduction}
Let
\[
  (\psi,A^{\mathrm A},L^{\mathrm A},A^{\mathrm B},L^{\mathrm B})
\]
be a projective strategy that passes the degree-\(d\) total-degree
diagonal line-vs-point test with probability at least \(1-\epsilon\),
where \(0\le\epsilon\le1\).
There is an invertible upper-triangular matrix \(T\in\GL_m(\F_q)\)
such that pulling back the strategy along \(y\mapsto Ty\), and padding
the diagonal-line answer space to degree at most \(md\), gives a
projective strategy for the \((m,q,d)\)-low individual degree test of
\cite[Section~3]{JNVWY20} with rejection probability at most
\[
  \rho:=\min\{1,3\epsilon/2\}.
\]
Its point measurements are
\(\widehat A^{r,y}_a=A^{r,Ty}_a\) for
\(r\in\{\mathrm A,\mathrm B\}\).
\end{lemma}

\begin{proof}
Let \(\eta_{\mathrm{AB}},\eta_{\mathrm{BA}},\delta\) be the rejection
probabilities of the Alice-point/Bob-line, Alice-line/Bob-point, and
point self-consistency subtests, respectively. Put
\[
  s:=\frac{\eta_{\mathrm{AB}}+\eta_{\mathrm{BA}}}{2},
  \qquad 2s+\delta\le3\epsilon.
\]

\medskip
\noindent\textbf{The transformed strategy.}
Put \(V_0=\{0\}\), and choose the columns \(Te_i\) independently and
uniformly from \(V_i\setminus V_{i-1}\). Then \(T\) is a uniformly random
invertible upper-triangular matrix, and \(T(V_i)=V_i\) for every \(i\).
Set \(\widehat A^{r,y}_a:=A^{r,Ty}_a\).

For an axis-parallel line \(\widehat\ell\) in the transformed
coordinates and an answer \(\widehat h\) of degree at most \(d\), let
\(h:T\widehat\ell\to\F_q\) be given by \(h(Ty)=\widehat h(y)\), and set
\[
  \widehat B^{r,\widehat\ell}_{\widehat h}
  :=
  L^{r,T\widehat\ell}_{h}.
\]

For the diagonal/general-line measurement of the low-individual-degree
test, whose answer space has degree bound \(md\), define, for every
\(\widehat\ell\), including singletons,
\[
  \widehat L^{r,\widehat\ell}_{\widehat h}
  :=
  \begin{cases}
  L^{r,T\widehat\ell}_{h}, & \deg \widehat h\le d,\\
  0, & d<\deg \widehat h\le md,
  \end{cases}
\]
again with \(h(Ty)=\widehat h(y)\). Thus degree-\(\le d\) total-degree
line answers are padded into the larger degree-\(\le md\) diagonal-line
answer space by assigning zero operator to the extra outcomes.

On a singleton every function has degree zero, so this same definition
pulls back the original singleton measurement. This agrees with the
literal sampler in \cite[Figure~2]{JNVWY20}, which allows \(\bv=0\).
Relabeling outcomes and adding zero projections preserve completeness
and projectivity.

\medskip
\noindent\textbf{The subtest distributions.}
Let \(\widehat\eta_{\mathrm{axis}}(T)\) and
\(\widehat\eta_{\mathrm{diag}}(T)\) denote the transformed strategy's
axis-parallel and diagonal subtest rejection probabilities, respectively,
each averaged over the two role orientations.
In the low-individual-degree axis-parallel subtest, the verifier samples
\(\by\sim\F_q^m\) and \(\bi\sim\{1,\ldots,m\}\) uniformly, and asks
the line
\[
  \widehat\ell_{\by,\bi}
  :=
  \{\by+t e_{\bi}:t\in\F_q\}.
\]
For any function
\(F:\F_q^m\times(\F_q^m\setminus\{0\})\to\mathbb R\), we have
\[
  \E_{T,\by,\bi}F(T\by,Te_{\bi})
  =
  \E_{\DiagPt}\bigl[F(\bu,\bv)\mid v_{\bi}\ne0\bigr].
\]
Indeed, conditional on any \(T\), the point \(T\by\) is uniform, so it
is independent of \(T\). For each \(i\), the column \(Te_i\) is uniform
on \(V_i\setminus V_{i-1}\), which is precisely the distribution of
\(\bv\sim V_i\) conditioned on \(v_i\ne0\). This event has probability
\(1-1/q\) for every \(i\), so conditioning leaves \(\bi\) uniform.

Since \(T\widehat\ell_{\by,\bi}=\{T\by+tTe_{\bi}:t\in\F_q\}\),
after averaging over \(T\) and the two role orientations, the
axis-parallel rejection probability is the original diagonal rejection
conditioned on \(v_{\bi}\ne0\). Nonnegativity gives
\[
  \E_T\widehat\eta_{\mathrm{axis}}(T)
  \le \frac{s}{1-1/q}.
\]

The transformed point self-consistency error is exactly \(\delta\), since
\(\by\mapsto T\by\) is a bijection of \(\F_q^m\).
For the diagonal/general-line subtest, for each fixed \(T\) and \(i\),
the map
\[
  (\by,\bv)\mapsto(T\by,T\bv)
\]
is a bijection of \(\F_q^m\times V_i\), since \(T(V_i)=V_i\).
It therefore preserves the independent uniform point and direction,
including the event \(\bv=0\). All line measurements, including singleton
measurements, were pulled back from the original strategy. Consequently,
the diagonal rejection probability is unchanged:
\[
  \widehat\eta_{\mathrm{diag}}(T)=s.
\]
\medskip
\noindent\textbf{Choosing a coordinate change.}
Writing \(\widehat\rho(T)\) for the transformed strategy's rejection
probability, the three equally weighted subtests give
\[
\begin{aligned}
  \E_T\widehat\rho(T)
  &=
  \frac13\E_T\widehat\eta_{\mathrm{axis}}(T)
  +\frac13\delta+\frac13s\\
  &\le \frac13\left(\frac{s}{1-1/q}+\delta+s\right)\\
  &\le \frac{3s+\delta}{3}
  \le \frac{2s+\delta}{2}
  \le \frac32\epsilon.
\end{aligned}
\]
Here we used \(q\ge2\) and \(\delta\ge0\).
Thus there exists a fixed upper-triangular \(T\in\GL_m(\F_q)\) for which
\(\widehat\rho(T)\le\rho\), as required.
\end{proof}

\begin{corollary}[Preliminary polynomial decoders]
\label[corollary]{cor:preliminary-decoders}
Let \((\psi,A^{\mathrm A},L^{\mathrm A},A^{\mathrm B},L^{\mathrm B})\)
be a projective strategy that passes the degree-\(d\) total-degree
diagonal line-vs-point test with probability at least \(1-\epsilon\),
where \(0\le\epsilon\le1\). Suppose \(q>md\), and let \(k\ge md\).
Put
\[
  K:=\max\{1,\lceil k\rceil\},
  \qquad \rho:=\min\{1,3\epsilon/2\},
  \qquad c_0:=\frac1{40000},
\]
and define
\begin{equation}
\label{eq:preliminary-decoder-error}
  \nu_{\mathrm{ind}}
  :=
  100000 K^2m^4
  \left(
    \rho^{c_0}
    +
    \left(\frac{d}{q}\right)^{c_0}
    +
    e^{-K/(2560000m^2)}
  \right).
\end{equation}
Then there exist projective measurements
\[
  \widetilde G^{\mathrm A},\widetilde G^{\mathrm B}
  \in\TotPolyMeas(m,q,md)
\]
such that, on average over \(\bu\sim\F_q^m\),
\[
\begin{aligned}
  A^{\mathrm A,u}_a\otimes I
  &\simeq_{\nu_{\mathrm{ind}}}
  I\otimes\widetilde G^{\mathrm B}_{[P(u)=a]},\\
  I\otimes A^{\mathrm B,u}_a
  &\simeq_{\nu_{\mathrm{ind}}}
  \widetilde G^{\mathrm A}_{[P(u)=a]}\otimes I,
\end{aligned}
\]
and their polynomial labels satisfy
\[
  \widetilde G^{\mathrm A}_P\otimes I
  \simeq_{\nu_{\mathrm{ind}}}
  I\otimes\widetilde G^{\mathrm B}_P.
\]
\end{corollary}

\begin{proof}
Use \Cref{lem:random-coordinate-reduction} to choose \(T\), and apply
\cite[Theorem~3.10]{JNVWY20} to the transformed strategy with the positive
integer parameter \(K\ge md\). It gives projective measurements
\(\widehat G^{\mathrm A},\widehat G^{\mathrm B}\in\polymeas{m}{q}{d}\)
whose two point-consistency errors and polynomial-label inconsistency
are at most \(\nu_{\mathrm{ind}}\).

For each individual-degree polynomial \(\widehat P\), define
\[
  P(x)=\widehat P(T^{-1}x).
\]
Since \(\widehat P\) has individual degree at most \(d\), it has total
degree at most \(md\), and composition with the invertible linear map
\(T^{-1}\) does not increase total degree. Thus \(P\) has total degree at
most \(md\). Since \(md<q\), each monomial has individual degree less
than \(q\), so this polynomial is already its reduced representative.

Define the measurements in the original coordinates by relabeling:
\[
  \widetilde G^r_P
  :=
  \sum_{\widehat P:\ \widehat P\circ T^{-1}=P}
  \widehat G^r_{\widehat P},
  \qquad
  r\in\{\mathrm A,\mathrm B\}.
\]
Since the map \(\widehat P\mapsto\widehat P\circ T^{-1}\) is a
deterministic relabeling of functions, this grouping preserves
completeness, projectivity, and polynomial-label consistency. Finally,
\(P(Ty)=\widehat P(y)\),
\(\widehat A^{r,y}_a=A^{r,Ty}_a\), and the bijection \(y\mapsto Ty\)
preserves the uniform point distribution. The two point-consistency
relations therefore hold in the original coordinates with the same error.
\end{proof}

\section{Controlling total degree}

\begin{lemma}[The uniform-line slice of the diagonal distribution]
\label[lemma]{lem:uniform-line-slice}
Let
\[
  \alpha
  :=
  \Pr_{\DiagPt}[\bi=m,\ \bv\ne0]
  =
  \frac{1-q^{-m}}{m}.
\]
Then, conditioned on the event \(\{\bi=m,\bv\ne0\}\), the induced
distribution of \((\bu,\ell)\) is exactly \(\LinePt\).
\end{lemma}

\begin{proof}
On the event \(\{\bi=m,\bv\ne0\}\), the vector \(\bv\) is uniform over
\(\F_q^m\setminus\{0\}\), independently of the uniform point \(\bu\).
Thus \(\ell=\{\bu+t\bv:t\in\F_q\}\) is generated by a uniform point and an
independent uniform nonzero direction.

For every fixed \(\bu\), each nondegenerate affine line through \(\bu\)
has exactly \(q-1\) nonzero direction vectors. Therefore the induced line
through \(\bu\) is uniform among the
\[
  \frac{q^m-1}{q-1}
\]
nondegenerate affine lines through \(\bu\). This is exactly \(\LinePt\).
\end{proof}

\begin{corollary}[Uniform-line domination]
\label[corollary]{cor:diag-dominates-uniform-lines}
For every nonnegative functional \(F(\bu,\ell)\) on nondegenerate
point-line pairs,
\[
  \E_{(\bu,\ell)\sim\LinePt}F(\bu,\ell)
  \le
  \alpha^{-1}
  \E_{(\bu,\bi,\bv,\ell)\sim\DiagPt}
  F(\bu,\ell)\mathbf 1_{\bv\ne0},
\]
where \(F(\bu,\ell)\mathbf 1_{\bv\ne0}\) is defined to be zero when \(\bv=0\).
\end{corollary}

\begin{proof}
Write \(A=\{\bi=m,\ \bv\ne0\}\), and extend \(F\) by zero to singleton
lines. By \Cref{lem:uniform-line-slice} and \(\Pr_{\DiagPt}[A]=\alpha\),
\[
  \begin{aligned}
  \E_{(\bu,\ell)\sim\LinePt}F(\bu,\ell)
  &=
  \E_{(\bu,\bi,\bv,\ell)\sim\DiagPt}
  \bigl[F(\bu,\ell)\mid A\bigr]\\
  &=
  \frac{1}{\alpha}
  \E_{(\bu,\bi,\bv,\ell)\sim\DiagPt}
  \bigl[F(\bu,\ell)\mathbf 1_A\bigr]\\
  &\le
  \frac{1}{\alpha}
  \E_{(\bu,\bi,\bv,\ell)\sim\DiagPt}
  \bigl[F(\bu,\ell)\mathbf 1_{\bv\ne0}\bigr].
  \end{aligned}
\]
The inequality uses \(F\ge0\) and \(A\subseteq\{\bv\ne0\}\).
For any nonnegative extension of \(F\) to singleton lines, the last
expression is also at most \(\alpha^{-1}\E_{\DiagPt}F(\bu,\ell)\).
In particular, since \(q\ge2\) and \(m\ge1\),
\[
  \alpha^{-1}
  =
  \frac{m}{1-q^{-m}}
  \le 2m.
\]
\end{proof}

\begin{lemma}[From point consistency to line consistency]
\label[lemma]{lem:point-to-line-consistency}
Let
\[
  (\psi,A^{\mathrm A},L^{\mathrm A},A^{\mathrm B},L^{\mathrm B})
\]
be a projective strategy that passes the degree-\(d\) total-degree
diagonal line-vs-point test with probability at least \(1-\epsilon\),
where \(0\le\epsilon\le1\).
Let \(\widetilde G^{\mathrm A},\widetilde G^{\mathrm B}\) be projective
measurements indexed by functions \(P:\F_q^m\to\F_q\), and write
\[
  \widetilde G^{r,u}
  :=\{\widetilde G^r_{[P(u)=a]}\}_{a\in\F_q}.
\]
Suppose, for some \(\nu\ge0\), that on average over \(\bu\sim\F_q^m\),
\[
\begin{aligned}
  A^{\mathrm A,u}_a\otimes I
  &\simeq_\nu I\otimes\widetilde G^{\mathrm B,u}_a,\\
  I\otimes A^{\mathrm B,u}_a
  &\simeq_\nu\widetilde G^{\mathrm A,u}_a\otimes I.
\end{aligned}
\]
Define the uniform-line inconsistencies by
\begin{equation}
\label{eq:decoder-line-errors}
\begin{aligned}
  \tau_{\mathrm A}
  &:=\E_{(\bu,\ell)\sim\LinePt}
    \Inc_\psi(\widetilde G^{\mathrm A,u},L^{\mathrm B,\ell,u}),\\
  \tau_{\mathrm B}
  &:=\E_{(\bu,\ell)\sim\LinePt}
    \Inc_\psi(L^{\mathrm A,\ell,u},\widetilde G^{\mathrm B,u}).
\end{aligned}
\end{equation}
Then
\[
  \tau_r\le3\nu+18m\epsilon,
  \qquad r\in\{\mathrm A,\mathrm B\}.
\]
\end{lemma}

\begin{proof}
Let \(\eta_{\mathrm{AB}},\eta_{\mathrm{BA}},\delta\) be the rejection
probabilities of the Alice-point/Bob-line, Alice-line/Bob-point, and
point self-consistency subtests. Thus
\(\eta_{\mathrm{AB}}+\eta_{\mathrm{BA}}+\delta\le3\epsilon\).
Define the corresponding uniform-line errors by
\[
\begin{aligned}
  \eta^{\mathrm{all}}_{\mathrm{AB}}
  &:=\E_{(\bu,\ell)\sim\LinePt}
    \Inc_\psi(A^{\mathrm A,u},L^{\mathrm B,\ell,u}),\\
  \eta^{\mathrm{all}}_{\mathrm{BA}}
  &:=\E_{(\bu,\ell)\sim\LinePt}
    \Inc_\psi(L^{\mathrm A,\ell,u},A^{\mathrm B,u}).
\end{aligned}
\]
By \Cref{cor:diag-dominates-uniform-lines}, nonnegativity of the
singleton-line contributions, and \(\alpha^{-1}\le2m\),
\[
  \eta^{\mathrm{all}}_o\le2m\eta_o,
  \qquad o\in\{\mathrm{AB},\mathrm{BA}\}.
\]
The marginal distribution of \(\bu\) under \(\LinePt\) is uniform.
We may therefore adjoin a line sampled conditionally on \(\bu\) to a
consistency relation averaged only over \(\bu\), without changing its error.

For projective measurements \(M,N\) on opposite subsystems, expanding the
square and using completeness gives
\[
\begin{aligned}
  &\sum_a\|(M_a\otimes I-I\otimes N_a)\ket{\psi}\|^2\\
  &\qquad=2-2\sum_a\bra{\psi}M_a\otimes N_a\ket{\psi}
  =2\Inc_\psi(M,N).
\end{aligned}
\]
This identity also holds after averaging over questions. All point-value
coarse-grainings here are projective, since they group orthogonal
projections. For Alice's decoder, telescope
\[
\begin{aligned}
  &\widetilde G^{\mathrm A,u}_a\otimes I
    -I\otimes L^{\mathrm B,\ell,u}_a\\
  &\quad=(\widetilde G^{\mathrm A,u}_a\otimes I
    -I\otimes A^{\mathrm B,u}_a)\\
  &\qquad+(I\otimes A^{\mathrm B,u}_a
    -A^{\mathrm A,u}_a\otimes I)\\
  &\qquad+(A^{\mathrm A,u}_a\otimes I
    -I\otimes L^{\mathrm B,\ell,u}_a).
\end{aligned}
\]
Apply \(\|z_1+z_2+z_3\|^2\le3\sum_{j=1}^3\|z_j\|^2\) to these
operators acting on \(\ket{\psi}\), sum over \(a\), and average over
\(\LinePt\). The projective identity gives
\[
  \tau_{\mathrm A}
  \le3\bigl(\nu+\delta
    +\eta^{\mathrm{all}}_{\mathrm{AB}}\bigr).
\]
For Bob's decoder, telescope through Alice's point measurement and then
Bob's point measurement, ending at Alice's line measurement. The same
calculation gives
\[
  \tau_{\mathrm B}
  \le3\bigl(\nu+\delta
    +\eta^{\mathrm{all}}_{\mathrm{BA}}\bigr).
\]
For either \(o\in\{\mathrm{AB},\mathrm{BA}\}\), the uniform-line
domination bound and the three-subtest error bound imply
\[
  \delta+\eta^{\mathrm{all}}_o
  \le2m(\delta+\eta_o)\le6m\epsilon.
\]
Consequently,
\[
  \tau_r\le3\nu+18m\epsilon,
  \qquad r\in\{\mathrm A,\mathrm B\}.
\]
\end{proof}

\begin{lemma}[Degree-collapse operator bound]
\label[lemma]{lem:degree-collapse-total-diag-proof}
Let \(P:\F_q^m\to\F_q\) be represented by a polynomial of total degree
\(D>d\), with
\[
  D\le md<q.
\]
Let
\[
  L^\ell=\{L^\ell_h\}_{\deg h\le d}
\]
be any projective degree-\(\le d\) line measurement on nondegenerate
affine lines. Then
\[
  \E_{(\bu,\ell)\sim\LinePt}
  L^\ell_{[h(u)=P(u)]}
  \le
  \frac{3D}{q}I
  \le
  \frac{3md}{q}I.
\]
\end{lemma}

\begin{proof}
Let \(P_D\) be the top homogeneous part of \(P\). For a parametrized line
\[
  a+tb,
  \qquad b\ne0,
\]
the coefficient of \(t^D\) in \(P(a+tb)\) is \(P_D(b)\).
Hence, if \(P_D(b)\ne0\), the restriction \(P|_\ell\) has univariate
degree exactly \(D>d\).

For such a line \(\ell\), and for any degree-\(\le d\) polynomial
\(h:\ell\to\F_q\), the function \(P|_\ell-h\) is a nonzero univariate
polynomial of degree at most \(D\). Since \(D<q\),
\[
  \Pr_{\bu\sim\ell}[h(\bu)=P(\bu)]
  \le
  \frac{D}{q}.
\]
Therefore, for every direction \(b\) with \(P_D(b)\ne0\),
\[
  \E_{\bu\sim\ell}
  L^\ell_{[h(u)=P(u)]}
  =
  \sum_{\deg h\le d}
  \Pr_{\bu\sim\ell}[h(\bu)=P(\bu)]L^\ell_h
  \le
  \frac{D}{q}I.
\]

It remains to bound the probability that \(P_D(b)=0\).
By Schwartz--Zippel,
\[
  \Pr_{b\sim\F_q^m}[P_D(b)=0]
  \le
  \frac{D}{q}.
\]
Conditioning on \(b\ne0\) increases this probability by at most a factor
\(2\), so
\[
  \Pr_{b\sim\F_q^m\setminus\{0\}}[P_D(b)=0]
  \le
  \frac{2D}{q}.
\]
On these bad directions we use the trivial bound by \(I\). Thus
\[
  \E_{(\bu,\ell)\sim\LinePt}
  L^\ell_{[h(u)=P(u)]}
  \le
  \left(\frac{2D}{q}+\frac{D}{q}\right)I
  =
  \frac{3D}{q}I.
\]
\end{proof}

\begin{lemma}[Pruning excessive-degree outcomes]
\label[lemma]{lem:degree-pruning}
Suppose \(q>6md\). Let
\((\psi,A^{\mathrm A},L^{\mathrm A},A^{\mathrm B},L^{\mathrm B})\)
be a projective strategy for the degree-\(d\) total-degree diagonal
line-vs-point test, and let
\(\widetilde G^{\mathrm A},\widetilde G^{\mathrm B}\in\TotPolyMeas(m,q,md)\).
Suppose their two cross-prover point-consistency errors and their
polynomial-label inconsistency are at most \(\nu\ge0\); explicitly,
on average over \(\bu\sim\F_q^m\),
\[
\begin{aligned}
  A^{\mathrm A,u}_a\otimes I
  &\simeq_\nu I\otimes\widetilde G^{\mathrm B}_{[P(u)=a]},\\
  I\otimes A^{\mathrm B,u}_a
  &\simeq_\nu\widetilde G^{\mathrm A}_{[P(u)=a]}\otimes I,
\end{aligned}
\]
and
\[
  \widetilde G^{\mathrm A}_P\otimes I
  \simeq_\nu I\otimes\widetilde G^{\mathrm B}_P.
\]
Let \(\tau_{\mathrm A},\tau_{\mathrm B}\) be their uniform-line
inconsistencies from \eqref{eq:decoder-line-errors}.
Then a common relabeling gives projective measurements
\(G^{\mathrm A},G^{\mathrm B}\in\TotPolyMeas(m,q,d)\)
whose two cross-prover point-consistency errors and polynomial-label
inconsistency are all at most
\[
  \nu+\frac{\max\{\tau_{\mathrm A},\tau_{\mathrm B}\}}{1-3md/q}.
\]
\end{lemma}

\begin{proof}
\textbf{Bounding the excessive-degree weight.}
Define
\[
  \Bad:=\{P:\ d<\Tot(P)\le md\},
  \qquad
  \widetilde G^r_{\Bad}:=\sum_{P\in\Bad}\widetilde G^r_P,
\]
and let
\[
  \mu_{\mathrm A}
  :=\bra{\psi}\widetilde G^{\mathrm A}_{\Bad}\otimes I\ket{\psi},
  \qquad
  \mu_{\mathrm B}
  :=\bra{\psi}I\otimes\widetilde G^{\mathrm B}_{\Bad}\ket{\psi}.
\]
These are the probabilities that the two decoders return polynomials of
total degree greater than \(d\). Put
\[
  \beta:=\frac{3md}{q}<\frac12.
\]
For every \(P\in\Bad\), completeness and
\Cref{lem:degree-collapse-total-diag-proof} give, on either prover's side,
\[
\begin{aligned}
  \E_{(\bu,\ell)\sim\LinePt}
    L^{r,\ell}_{[h(u)\ne P(u)]}
  &=I-\E_{(\bu,\ell)\sim\LinePt}
    L^{r,\ell}_{[h(u)=P(u)]}\\
  &\ge(1-\beta)I.
\end{aligned}
\]
Expanding the decoder's point-value coarse-graining and retaining only
the nonnegative bad-outcome terms now yields
\[
\begin{aligned}
  \tau_{\mathrm A}
  &=\E_{(\bu,\ell)\sim\LinePt}\sum_P
    \bra{\psi}\widetilde G^{\mathrm A}_P\otimes
    L^{\mathrm B,\ell}_{[h(u)\ne P(u)]}\ket{\psi}\\
  &\ge(1-\beta)\sum_{P\in\Bad}
    \bra{\psi}\widetilde G^{\mathrm A}_P\otimes I\ket{\psi}\\
  &=(1-\beta)\mu_{\mathrm A}.
\end{aligned}
\]
Here \(\widetilde G^{\mathrm A}_P\) is independent of the sampled questions,
and tensoring a positive-operator inequality with this positive projection
preserves its order. Applying the same argument on the other tensor
factor gives explicitly
\[
\begin{aligned}
  \tau_{\mathrm B}
  &=\E_{(\bu,\ell)\sim\LinePt}\sum_P
    \bra{\psi}L^{\mathrm A,\ell}_{[h(u)\ne P(u)]}\otimes
    \widetilde G^{\mathrm B}_P\ket{\psi}\\
  &\ge(1-\beta)\sum_{P\in\Bad}
    \bra{\psi}I\otimes\widetilde G^{\mathrm B}_P\ket{\psi}\\
  &=(1-\beta)\mu_{\mathrm B}.
\end{aligned}
\]
Thus, for both provers,
\[
  \mu_r
  \le\frac{\tau_r}{1-\beta},
  \qquad r\in\{\mathrm A,\mathrm B\}.
\]

\medskip
\noindent\textbf{Relabeling the excessive-degree outcomes.}
Define the pruning map
\[
  \pi(P)
  :=
  \begin{cases}
  P, & \Tot(P)\le d,\\
  0, & \Tot(P)>d.
  \end{cases}
\]
Set, for \(r\in\{\mathrm A,\mathrm B\}\) and
\(g\in\TotPoly(m,q,d)\),
\[
  G^r_g
  :=
  \sum_{P:\pi(P)=g}\widetilde G^r_P.
\]
This is a deterministic coarse-graining of a projective measurement, so
\[
  G^{\mathrm A},G^{\mathrm B}\in\TotPolyMeas(m,q,d)
\]
are projective.

Self-consistency is preserved by deterministic post-processing:
\[
  G^{\mathrm A}_g\otimes I
  \simeq_\nu
  I\otimes G^{\mathrm B}_g.
\]

For point consistency, the scalar inequality
\[
  \mathbf 1_{a\ne\pi(P)(u)}
  \le\mathbf 1_{a\ne P(u)}+\mathbf 1_{P\in\Bad}
\]
shows that relabeling increases the error by at most the corresponding
bad-outcome mass. Therefore, on average over \(\bu\sim\F_q^m\),
\[
  A^{\mathrm A,u}_a\otimes I
  \simeq_{\nu+\mu_{\mathrm B}}
  I\otimes G^{\mathrm B}_{[g(u)=a]},
\]
and
\[
  I\otimes A^{\mathrm B,u}_a
  \simeq_{\nu+\mu_{\mathrm A}}
  G^{\mathrm A}_{[g(u)=a]}\otimes I.
\]
All three errors are therefore at most
\[
  \nu+\max\{\mu_{\mathrm A},\mu_{\mathrm B}\}
  \le\nu+\frac{\max\{\tau_{\mathrm A},\tau_{\mathrm B}\}}{1-\beta}.
\]
\end{proof}

\section{Proof of the soundness theorem}
\label{sec:total-degree-proof}

\begin{proof}[Proof of \Cref{thm:total-degree-diagonal-from-individual-degree}]
Take \(C_0=2000000\). If \(q\le6md\), use the trivial projective decoders
\(G^{\mathrm A}_0=G^{\mathrm B}_0=I\), with every other outcome zero.
Their polynomial labels agree perfectly, and each point-consistency
error is at most \(1\). Since
\[
  (k+1)^2m^4\left(\frac{d}{q}\right)^{c_0}
  \ge6^{-c_0}(k+1)^2m^{4-c_0}
  \ge6^{-c_0},
\]
the claimed bound holds in this case.

Now assume \(q>6md\). Apply \Cref{cor:preliminary-decoders} to obtain
projective decoders \(\widetilde G^{\mathrm A},\widetilde G^{\mathrm B}\)
of total degree at most \(md\),
with error \(\nu_{\mathrm{ind}}\) from \eqref{eq:preliminary-decoder-error}.
By \Cref{lem:point-to-line-consistency},
\(\tau_r\le3\nu_{\mathrm{ind}}+18m\epsilon\) for both provers.
Since \(1-3md/q>1/2\), \Cref{lem:degree-pruning} gives degree-\(\le d\)
projective decoders whose three consistency errors are at most
\[
  \nu_{\mathrm{ind}}
  +2(3\nu_{\mathrm{ind}}+18m\epsilon)
  =7\nu_{\mathrm{ind}}+36m\epsilon.
\]

It remains to bound this expression by \(\nu_{\mathrm{tot}}\) as defined
in the theorem. With \(K=\max\{1,\lceil k\rceil\}\) and
\(\rho=\min\{1,3\epsilon/2\}\), we have
\[
\begin{aligned}
  \rho^{c_0}\le(3/2)^{c_0}\epsilon^{c_0},
  &\qquad K^2\le(k+1)^2,\\
  e^{-K/(2560000m^2)}
  &\le e^{-k/(2560000m^2)}.
\end{aligned}
\]
Also \(m\epsilon\le(k+1)^2m^4\epsilon^{c_0}\), since
\(0<c_0<1\), \(0\le\epsilon\le1\), and \(m\ge1\).
Substituting into \eqref{eq:preliminary-decoder-error} gives
\[
\begin{aligned}
  7\nu_{\mathrm{ind}}+36m\epsilon
  &\le (k+1)^2m^4\Bigl[
    \bigl(700000(3/2)^{c_0}+36\bigr)\epsilon^{c_0}\\
  &\qquad\qquad+700000\left(\frac{d}{q}\right)^{c_0}
    +700000e^{-k/(2560000m^2)}\Bigr]\\
  &\le\nu_{\mathrm{tot}},
\end{aligned}
\]
because \(700000(3/2)^{c_0}+36<2000000=C_0\).
\end{proof}

\paragraph{AI usage.}
ChatGPT 5.5 first found this valid proof, and ChatGPT 6 was used in the
writing of this paper.

\end{document}